\documentclass[final,12pt]{arxivalt2021}

\LinesNumbered
\usepackage{hyperref}
\usepackage{theoremref}
\usepackage[utf8]{inputenc}
\usepackage{amsfonts, amssymb, mathtools,enumitem,bbold}

\usepackage{thmtools}
\usepackage{thm-restate}
\usepackage{caption}
\usepackage{nicefrac}
\usepackage{capt-of}
\usepackage{cases}
\usepackage[capitalise]{cleveref}
\usepackage{color}

 \newif\ifnotes\notesfalse

\ifnotes
\definecolor{mygrey}{gray}{0.50}
\newcommand{\notename}[2]{{\textcolor{red}{\footnotesize{\bf (#1:} {#2}{\bf ) }}}}

\newcommand{\annote}[1]{{\color{red}Anindya: #1}}
\newcommand{\annnote}[1]{{\color{blue}Anindya: #1}}
\newcommand{\chnote}[1]{{\color{brown} #1}}
\newcommand{\newd}[1]{{\color{violet}#1}}
\newcommand{\newa}[1]{{\color{purple} #1}}
\newcommand{\cadnote}[1]{{\color{olive} #1}}

\else

\newcommand{\notename}[2]{{}}

\newcommand{\annote}[1]{}
\newcommand{\annnote}[1]{}
\newcommand{\chnote}[1]{#1}
\newcommand{\newd}[1]{#1}
\newcommand{\newa}[1]{#1}
\newcommand{\cadnote}[1]{#1}

\fi

\newcommand{\anote}[1]{{\notename{Aravindan}{#1}}}

\newcommand{\cnote}[1]{{\notename{Aidao}{#1}}}

\newtheorem*{theorem*}{Theorem}

\newtheorem*{hypothesis*}{Hypothesis}

\newtheorem{claim}[theorem]{Claim}
\newtheorem*{claim*}{Claim}

\newtheorem*{observation*}{Observation}

\newcommand{\eat}[1]{}

\newcommand{\mathset}[1]{\{ #1 \} }
\newcommand{\setsize}[1]{ | #1 | }
\newcommand{\poly}{\mathsf{poly}}

\newcommand{\bbF}{\mathbb{F}_2}
\newcommand{\R}{\mathbb{R}}
\newcommand{\bbP}{\mathbb{P}}
\newcommand{\iprod}[1]{\langle#1\rangle}

\newcommand{\Amax}{A_{0}}
\newcommand{\Amin}{A_{1}}
\newcommand{\dmax}{d_{0}}
\newcommand{\dmin}{d_{1}}
\newcommand{\wfirst}{w_0}
\newcommand{\wsecond}{w_1}
\newcommand{\wnaught}{{w_{min}}}
\newcommand{\Ifirst}{I_0}
\newcommand{\Isecond}{I_1}
\newcommand{\WLOG}{Without loss of generality}
\newcommand{\lowerwlog}{without loss of generality}
\newcommand{\sampleo}{\chnote{\mathcal{O}(\Amax,\Amin,\wfirst,\wsecond)}}
\newcommand{\sampleouv}{\chnote{\mathcal{O}(U, V, w_U, w_V)}}
\newcommand{\Dsampleouv}{\chnote{\mathcal{O}(DU, DV, w_U, w_V)}}

\newcommand{\alglargediff}{\textsc{Subspace-Recover-Large-Diff}}
\newcommand{\algtestcompr}{\textsc{Test-Comparability}}
\newcommand{\algincprrcv}{\textsc{Incomparable-Subspace-Recovery}}
\newcommand{\algfindagoodprojector}{\textsc{Find-A-Good-Projector}}
\newcommand{\algchoosetherighthypo}{\textsc{Choose-The-Right-Hypothesis}}

\newcommand{\zero}{\mathsf{zero}}
\newcommand{\rank}{\mathsf{rank}}
\newcommand{\bx}{\mathbf{x}}
\newcommand{\by}{\mathbf{y}}
\newcommand{\bz}{\mathbf{z}}
\newcommand{\bZ}{\chnote{{\mathbf{Z}}}}
\newcommand{\bD}{\chnote{{\mathbf{D}}}}

\newcommand{\penclose}[1]{\left( #1 \right) }

\newcommand{\bfM}{\mathbf{M}}
\newcommand{\bfT}{\mathbf{T}}
\newcommand{\bft}{\mathbf{t}}
\newcommand{\bfy}{\mathbf{y}}
\newcommand{\bfx}{\mathbf{x}}
\newcommand{\bfb}{\mathbf{b}}

\newcommand{\sfspan}{\mathsf{span}}

\newcommand{\calE}{\mathcal{E}}
\newcommand{\qqquad}{\qquad\qquad\qquad\qquad\qquad\qquad\qquad\qquad\qquad}

\newcommand{\True}{\ensuremath{\mathsf{True}}}
\newcommand{\False}{\ensuremath{\mathsf{False}}}
\newcommand{\sfRM}{\ensuremath{\mathsf{RM}}}

\DeclarePairedDelimiter\ceil{\lceil}{\rceil}

\newcommand{\spann}{\mathsf{span}}
\newenvironment{proofof}[1]{\noindent {\bf Proof
of #1.} }
{\hfill $\blacksquare$ \bigskip}

\title{Learning a mixture of two subspaces over finite fields}

\altauthor{%
 \Name{Aidao Chen}\footnotemark[2]  \Email{aidaochen2022@u.northwestern.edu}\\
 \addr Northwestern University
 \AND
 \Name{Anindya De}\thanks{Supported by NSF grants CCF 1910534 and CCF 1926872. Part of the work was done while visiting the Simons Institute for Theory of Computing for the program ``Probability, Geometry and Computation in High Dimensions".} \Email{anindyad@seas.upenn.edu}\\
 \addr University of Pennsylvania%
  \AND
 \Name{Aravindan Vijayaraghavan}\thanks{Supported by NSF grants CCF-1652491, CCF-1637585 and CCF-1934931.}\Email{aravindv@northwestern.edu}\\
 \addr Northwestern University. 
}
\date{}

\begin{document}

\maketitle

\thispagestyle{empty}

\begin{abstract}
  We study the problem of learning a mixture of two subspaces over \newd{$\bbF^n$}. 
  The goal is to recover  the individual subspaces $A_0, A_1$, 
  given samples from a (weighted) mixture of samples drawn uniformly from the subspaces $A_0$ and $A_1$. 
  This problem is computationally challenging, as it captures the notorious problem of ``learning parities with noise" 
  in the degenerate setting when $\Amin \subseteq \Amax$. 
  This is in contrast to the analogous problem over the reals that can be solved in polynomial time (Vidal'03). This leads to the following natural question: 
  \newd{{\em is Learning Parities with Noise the only computational barrier in obtaining efficient algorithms for learning mixtures of subspaces over $\bbF^n$?}}   
  
  The main result of this paper is an affirmative answer to the above question. Namely, we show the following results: 
  \begin{enumerate}
      \item When the subspaces $A_0$ and $A_1$ are \emph{incomparable}, i.e., $A_0 \not \subseteq A_1$ and $A_1 \not \subseteq A_0$, then there is a polynomial time algorithm to recover the subspaces $A_0$ and $A_1$. 
      \item \newa{In the case when $\Amin \subseteq \Amax$ such that $\mathsf{dim}(\Amin) \le \alpha \cdot \mathsf{dim}(\Amax)$ for $\alpha<1$, there is a $n^{O(1/(1-\alpha))}$ time algorithm to recover the subspaces $\Amax$ and $\Amin$.} 
  \end{enumerate}
 
 Thus, our algorithms imply computational tractability of the problem of learning mixtures of two subspaces, except in the degenerate setting captured by learning parities with noise. 
 
  

  

  \end{abstract}


\section{Introduction}
Mixture models form an expressive class of probabilistic models that are widely used to find structure in unlabeled data from a heterogeneous population. Each of the $k$ components in a mixture model represents one of the $k$ sub-populations \newd{(assumed to be homogeneous)} that constitute the overall heterogeneous population.  
A variety of mixture models ranging from Gaussian mixture models and mixtures of product distributions over continuous domains, to mixtures of ranking models, mixtures of subcubes over discrete domains are used to capture data in different domains. There is an extensive literature in statistics and computer science that gives efficient polynomial time algorithms for learning many mixture models with a constant number of mixture components~\citep{FOS06,KMV10, MV10, BS10, RSS, LRSS15, ABSheffet, LM18, CM19}. 

A common assumption in high-dimensional data analysis is to assume that the given data belong to a collection of lower dimensional subspaces.
A prominent line of work in machine learning, computer vision and computational geometry~\citep{Vidal, EV, CandesE,Sanghavi} that formalizes this intuition is the problem of learning a mixture of subspaces (or subspace clustering).  Given a set of points in $n$ dimensions that belong to a union of $k\ge 2$ subspaces,  the goal is to find the individual subspaces that contain all the points.
When the points belong to $\R^n$, a beautiful result of \citet{Vidal} shows that for any mixture of $k$ subspaces, under some mild general-position assumption of the points in the subspaces,\footnote{Such an assumption is necessary, to ensure that the individual subspaces are identifiable.} there is an algorithm that runs in time $n^{O(k)}$ that recovers the $k$ individual subspaces. Very recently, subspace clustering has also been studied with outlier noise, in the special case when the points in each cluster is drawn from a Gaussian supported on a subspace~\citep{RaghavendraYao20, BakshiKothari20}.  
However these guarantees are specific to the real domain. A natural question is whether such algorithmic guarantees also extend to other domains like $\bbF$. 

\vspace{5pt}
{\em Can we efficiently learn a mixture of subspaces over finite fields?}
\vspace{5pt}

The algorithmic problem has a very different flavor over finite fields and becomes computationally challenging even in simple settings. 
In the simplest setting, we are given samples from a mixture of $k=2$ unknown subspaces $A_0, A_1 \subseteq \bbF^n$ \cadnote{of dimension $d_0,d_1$ (respectively)}, with unknown mixing weights $w_0, w_1 \in [0,1]$ that add up to $1$. Each sample is drawn independently as follows: with probability $w_0$, the sample is drawn from $\mathsf{U}_{A_0}$, the uniform distribution over subspace $A_0 \subseteq \bbF^n$, and with $w_1$ the sample is drawn from the uniform distribution $\mathsf{U}_{A_1}$ over $A_1 \subseteq \bbF^n$. The goal is to learn the individual subspaces $A_0, A_1$ from independent samples generated from this model.  We refer the reader to Definition~\ref{def:model} for the formal definition of the model. 

\newd{Learning mixtures of subspaces over $\bbF$ essentially generalizes the problem of learning mixtures of subcubes that was studied in \citep{CM19}. In particular, subcubes correspond to (affine) subspaces where the constraints are given by  standard unit vectors. On the other hand, in this work, we consider arbitrary subspaces of $\bbF^n$ (though we do not allow for affine subspaces).
Our work can also be through the framework of  \emph{learning from positive examples}~\cite{denis2005learning, de2014learning, canonne2020learning, ernst2015algorithmic} which studies the learnability of supervised concept classes (in this case subspaces) when the algorithm only gets positive samples.}

More interestingly, the simple setting of $k=2$ already captures the notorious problem of learning parities with noise (LPN) as a special case. One can encode LPN as learning a mixture of two subspaces $A_0, A_1$ where the subspaces $A_1 \subset A_0 \subseteq \bbF^n$ and $\text{dim}(A_1) = \text{dim}(A_0)-1$ (see Proposition~\ref{prop:reduce_mix_to_LPN} and Proposition~\ref{prop:reduce_LPN_to_mix}). 
The best known algorithm for LPN runs in time $\exp\big(O(n/\log n)\big)$ \citep{BKW}. Moreover LPN is also used as an average-case hardness assumption in learning theory and cryptography ~\citep{LPNsurvey}. To avoid this computational barrier, we will assume that we are not in the degenerate setting when one subspace contains the other.  We call the two subspaces $A_0$ and $A_1$ {\em incomparable} iff $A_0 \nsubseteq A_1$ and $A_1\nsubseteq A_0$. This leads to the following natural  question about the computational complexity of the problem:

 \vspace{5pt}
 \noindent {\bf Question.} {\em Is LPN the only computational obstruction for learning a mixture of two subspaces? Can one design faster algorithms when the subspaces $A_0, A_1$ are incomparable?}
\vspace{5pt}

Our first
result shows that one can indeed design a polynomial time algorithm when the two subspaces are incomparable.
\begin{restatable}{theorem}{restatefastincprrecovery}\label{thm:intro:main}\label{thm:incpr_rcv}
There is an algorithm \algincprrcv{} with the following guarantee: given oracle access to $\sampleo$ (for unknown $\Amax,\Amin,\wfirst,\wsecond$), $\wnaught> 0$ (such that $\wnaught\le\min\mathset{\wfirst,\wsecond}$) and confidence parameter $\delta>0$,
\begin{enumerate}
    \item \algincprrcv{} runs in sample and time complexity $\poly(n/\wnaught)\cdot\log(1/\delta)$
    \item With probability $1-\delta$, the algorithm outputs the subspaces $\Amax,\Amin$, and estimates the weights $w_0, w_1$ up to any desired inverse polynomial accuracy.
\end{enumerate}
\end{restatable}

Hence the above result gives a significantly faster polynomial
time algorithm if we are {\em not} in the degenerate {\em comparable} setting when one subspace contains the other. In contrast, when \newd{$A_1 \subset A_0$ and $dim(A_1)=dim(A
_0)-1$ (or vice versa)},  the best known algorithm takes $\exp({O}(n/\log n))$ time. We remark that the algorithm succeeds in uniquely identifying and recovering the individual subspaces, as opposed to just finding a mixture of two subspaces that fits the data. In the parlance of statistics, our algorithm recovers the underlying model (sometimes referred to as \emph{parameter estimation}) as opposed to just doing \emph{density estimation}.

Next, observe that the (presumed) hardness of LPN only implies hardness of the subspace recovery problem when (i) $A_1 \subseteq A_0$  and (ii) $\mathsf{dim}(A_1) = \mathsf{dim}(A_0) -1$. This naturally prompts the question whether subspace recovery remains hard if (say) $A_1 \subseteq A_0$ but $\mathsf{dim}(A_1)\ll \mathsf{dim}(A_0)$. In other words, we ask the following question: 
 
 \vspace{5pt}
 \noindent {\bf Question.} {\em Can we design fast algorithms for subspace recovery when $\mathsf{dim}(A_0)$ and $\mathsf{dim}(A_1)$ are substantially different?
Note that we are not imposing any conditions on the comparability of the hidden subspaces $A_0$ and $A_1$. }
\vspace{5pt}

Our next result provides an affirmative answer to this question. 
\begin{restatable}{theorem}{restatedversusalphad}\label{thm:intro:dversusalphad}\label{thm:dversusalphad}
Let  $\wnaught \ge 1/100$.  
Let $d_0 \ge d_1$ and suppose $\alpha:=\dmin/\dmax < 1- \frac{\log \dmax}{\sqrt{\dmax}}$. There is an algorithm \alglargediff{} with the following guarantee: given oracle access to $\sampleo$(for unknown $\Amax,\Amin,\wfirst,\wsecond)$, $\wnaught> 0$ (such that $\wnaught\le \min\mathset{\wfirst,\wsecond}$) and confidence parameter $\delta>0$,
\begin{enumerate}
    \item \alglargediff{} runs in sample and time complexity
    \newline 
    $\log(1/\delta)\poly(n)\cdot \dmax^{O(1)/(1-\alpha)}$. 
    \item With probability $1-\delta$, the algorithm outputs the subspaces $\Amax,\Amin$, and estimates the mixing weights up to any desired inverse polynomial accuracy.
\end{enumerate}
\end{restatable}

Informally speaking, 
if the ratio of dimensions $\alpha$ is bounded away from $1$, the running time is polynomial. In general, the running time of the algorithm has a dependence of $O(1/(1-\alpha))$ in the exponent. 

\subsection{Overview of Techniques.}

We now  briefly describe the algorithmic ideas and techniques used to prove our results.  
The algorithms that establish Theorem~\ref{thm:intro:main} and Theorem~\ref{thm:dversusalphad} use very different ideas. \newd{We begin with an overview of Theorem~\ref{thm:intro:main}.}

\paragraph{Incomparable Setting (Theorem~\ref{thm:intro:main}).}
The main component of the polynomial time algorithm in the incomparable setting is a careful procedure for {\em dimension reduction} that reduces the subspace clustering problem to $O(1)$ dimensions. We will construct a matrix $M \in \bbF^{r \times n}$ where $r=O(1)$ \newd{(in the actual proof, we set $r=10$)}, and solve the clustering problem given samples of the form $y=M x$ where $x$ is drawn from the original mixture. Note that a subspace under any linear map $M$ also gives a subspace; hence the samples in $\R^r$ are drawn from a mixture of subspaces $MA_0$ and $MA_1$. Any algorithm for learning a mixture of subspaces in  $r=O(1)$ dimensions will allow us to cluster the points, and recover the individual subspaces $\Amax, \Amin$. 

How do we choose the linear map $M$? \newd{A key property that we require of $M$ is that if $A_0$ and $A_1$ are incomparable, then $MA_0$ and $MA_1$ should also remain incomparable. While it is not hard to see that such a $M$ exists (even when $r=O(1)$), it is far from clear how to find it given that we do not have $A_0$ and $A_1$ explicitly. A natural choice for $M$ is a random matrix, where every entry is chosen independently from $\bbF$. Random linear maps are often used for dimension reduction in the real domain to approximately preserve inner products and pairwise distances. However, a random map does not work in our setting, particularly when the target dimension $r \ll \dmin$. This is because with high probability the subspaces collapse and $\bfM A_0 = \bfM A_1 = \bbF^r$, thereby making it impossible to recover the individual subspaces $\bfM A_0, \bfM A_1$.}

Our approach instead proceeds in multiple \newd{rounds}, where \newd{in each round, we reduce the dimension by one while preserving the property that the projected subspaces remain incomparable. More precisely,} one can show that for a random linear map $\bfM_{n-1} \in \bbF^{(n-1) \times n}$, with {\em constant probability}, \newd{$\bfM_{n-1} A_0$ and $\bfM_{n-1} A_1$  are incomparable} if $A_0, A_1$ are originally incomparable. However, this does not suffice per se, since we want to apply this for $\Omega(n)$ \newd{rounds} (and thus, the probability of success becomes exponentially small).  \newd{The crucial component of our algorithm is a testing procedure that runs in polynomial time, which given samples from a mixture of subspaces $U, V$,  w.h.p. outputs whether $U$ and $V$ are {\em comparable} or {\em incomparable}.} With such a procedure, in every phase we can reduce the dimension by $1$, by sampling several random linear maps, running our testing procedure on each of them, and picking one that preserves incomparability of the subspaces. 
The guarantee of the testing procedure is given below.

\begin{restatable}{theorem}{restatetestcomparability}\label{thm:test_comparability}
There is an algorithm \algtestcompr{} with the following guarantee: Given oracle access to $\mathcal{O}(U, V, w_U, w_V)$ (for unknown $U, V, w_U, w_V$), $\wnaught > 0$ (such that $ \min\{w_U, w_V\}\ge \wnaught$) and confidence parameter $\delta>0$,
\begin{enumerate}
    \item \textsc{Test-comparability} runs in sample and time complexity $1/\wnaught^2 \cdot \text{poly}(n)\log(1/\delta)$.
    \item With probability $1-\delta$, the algorithm outputs \True{} if $U$ and $V$ are comparable and \False{} otherwise.
\end{enumerate}
\end{restatable}

The testing procedure uses the following main insight. Suppose for simplicity the span $\mathsf{span}(U \cup V) = \bbF^n$. 
We prove that the subspaces $U$ and $V$ are incomparable if and only if there exists a non-zero polynomial $p$ of degree $2$ \newd{that vanishes on $\mathcal{A}=U \cup V$. In fact, it will suffice to choose $\mathcal{A}$ to be a randomly chosen set of polynomial size sampled from the mixture of subspaces $U$ and $V$.}
The set of feasible degree-$2$ polynomials can then be obtained by setting up a system of linear equations where the unknowns correspond to co-efficients of $p$.       

 \newd{Let us define $\bfM \in \bbF^{O(1) \times n}$ as $\bfM  = \bfM_{r} \cdot \bfM_{r+1} \cdot \ldots \cdot \bfM_{n-1}$ -- in other words, $\bfM$ is the linear map obtained by composing the dimension reduction maps over the $n-r$ rounds.} 
 Once the dimension is reduced to $r=O(1)$, we use a brute-force algorithm to recover $\bfM \Amax, \bfM \Amin$. Finally, once we know $\bfM\Amax,\bfM\Amin$, we can draw uniform samples from $\Amax\backslash\mathset{x\in \Amax:\bfM x\in \bfM \Amin}$ to recover $\Amax$; we can recover  $\Amin$ similarly (see Lemma~\ref{lem:minuspropersubspace_spanall}).


\paragraph{Significant dimension difference (Theorem~\ref{thm:dversusalphad}).}
When the dimension of the subspaces are substantially different, we use algebraic ideas inspired from techniques in the real domain to recover the subspaces. 
The main algorithmic idea is by adapting ideas from related problem of subspace recovery over the reals \citep{HM13, BCPV}. \newd{To explain the idea,}
consider the setting with equal mixing weights of $1/2$, $\dmax \approx n$, and suppose $\alpha = 1-\Omega(1)$. 
If we consider a random subsample of $\dmax$ points from the data set, we expect to have roughly $\dmax/2$ points from subspace $\Amax$ and $\dmax/2$ points from subspace $\Amin$. Suppose $\alpha < 1/2$ \newd{(referred to as the ``large gap case")}
 i.e.,   $\dmin< \dmax/2$, then with high probability there is a linear dependence in this sub-sample.
\newd{Further, this linear dependence is (entirely) among points lying in the subspace $\Amin$.
This can be used to recover the subspace $\Amin$ (and consequently, the subspace $\Amax$ as well).}

\newd{To see why this idea does not work in general, consider the case when the weights $w_0= 0.9, w_1=0.1$ and $d_1 = 0.8 d_0$. Then, to see a linear dependence among the points in $A_1$, we need to sample at least $d_1$ points from $A_1$. However, on an average, this will mean sampling around $(w_0/w_1) \cdot d_1 = 9d_1$ many points from $A_0$. As $9d_1$ is much larger than the ambient dimension and thus, we will find many \emph{spurious} linear dependencies -- i.e., dependencies which do not come from points belonging to $A_1$. Thus, this strategy will fail to identify $A_1$.} 

\newd{Instead, when $\alpha \ge 1/2$, we will adopt a \emph{dimension gap amplification strategy}. In particular, we consider a non-linear map $\phi:\bbF^{\dmax} \rightarrow \bbF^{\dmax'}$ where $\dmax' = \sum_{j=0}^\ell \binom{\dmax}{j}$
for an appropriately chosen $\ell$. Further, for a set $B$, let us define $\phi(B)$ as the set $\{\phi(x): x \in B\}$. 
Roughly speaking, we want to choose an appropriate $\ell$ such that $\mathsf{dim}(\spann(\phi(A_1))) / \mathsf{dim}(\spann(\phi(A_0))) <1/2$. For such an $\ell$, we can now apply the strategy for the large gap case to recover $A_1$ and $A_0$. We note that the idea of such a dimension gap amplification was also applied in the related subspace recovery problem over reals~\citep{BCPV} -- there, the goal was recover one subspace $S$ of dimension $d \le n$ containing $o(d/n)$ fraction of the points,
while the rest of the points are drawn in general position from the whole of $\R^n$. While in spirit our idea is similar, it is challenging to get a handle on the dimensions of $\spann(\phi(A_1))$ and $\spann(\phi(A_0))$. In particular, the techniques of \cite{BCPV} which are meant for the reals, do not seem to be applicable in the finite field setting. 
Fortunately for us, some powerful results from additive combinatorics ~\citep{keevash2005set, ben2012random} let us get precise estimates for $\mathsf{dim}(\spann(\phi(A_0)))$ and $\mathsf{dim}(\spann(\phi(A_1)))$. Roughly speaking, we show that for $\ell \approx 1/(1-\alpha)$, $\mathsf{dim}(\spann(\phi(A_1))) / \mathsf{dim}(\spann(\phi(A_0))) <1/2$, thus reducing to the large gap case. 
}

\section{Preliminaries} \label{sec:prelims}
We start by defining the subspace recovery problem formally. 
\begin{definition}~\label{def:model} 
The Subspace-Recovery problem is instantiated by  two subspaces of $\bbF^n$ - $A_0$ and $A_1$ of dimensions $d_0$ and $d_1$ respectively.
In addition, we also have weights 
 $\wfirst$ and $\wsecond$ such that $\wfirst + \wsecond=1$. 
 
The subspaces $A_0$, $A_1$, dimensions $d_0$, $d_1$ as well as the weights $\wfirst$ and $\wsecond$ are unknown. For this instance, we define the sampling oracle $\sampleo$ is defined as follows: sample $\mathbf{b} \in \{0,1\}$ where $\Pr[\mathbf{b}=0] = \wfirst$ and $\Pr[\mathbf{b}=1] = \wsecond$. If $\mathbf{b}=0$, $\sampleo$ outputs a uniformly random element from $\Amax$ and if $\mathbf{b}=1$, $\sampleo$ outputs a uniformly random element from $\Amin$.

In the Subspace-Recovery problem, the algorithm is given access to the sampling oracle $\sampleo$, an error parameter $\epsilon>0$ and a weight parameter $\wnaught>0$ with the promise  that $\wnaught \le \min \{\wfirst, \wsecond\}$.
The goal of the algorithm is to output subspaces  $\Amax, \Amin$ and estimates $\hat{w}_0$, $\hat{w}_1$ such that $|\wfirst -\hat{w}_0| + |\wsecond - \hat{w}_1| \le \epsilon$.

\chnote{\WLOG{}, we will assume $d_0\ge d_1$ from now on.}
\begin{remark} \label{rem:weights}
Note that once $\Amax,\Amin$ is found, estimating $\wfirst,\wsecond$ is not hard, this is because $\bbP_{\bx\sim\sampleo}[\bx\in\Amax\setminus\Amin]=\wfirst \frac{|\Amax\setminus\Amin|}{|\Amax|}$. 
Formally, there is an algorithm with the following guarantee: given oracle access to $\sampleo$ (for unknown $\wfirst,\wsecond$), $\Amax,\Amin$ and confidence parameter $\delta>0$,
\begin{enumerate}
    \item this algorithm runs in sample and time complexity $\poly(n)\cdot1/\epsilon^2\cdot\log(1/\delta)$
    \item With probability $1-\delta$, the algorithm outputs $\hat{w}_0$, $\hat{w}_1$ such that $|\wfirst -\hat{w}_0| + |\wsecond - \hat{w}_1| \le \epsilon$.
\end{enumerate}
By this observation, we can focus on finding $\Amax,\Amin$ from now on.

\end{remark}
\end{definition}
We next define the concept of incomparable subspaces. 
\begin{definition}
We define two subspaces $A,B$ to be  incomparable if and only if $A\nsubseteq B$ and $B\nsubseteq A$.
\end{definition}
\subsubsection{Some useful notation}
\begin{enumerate}
    \item For any $f: \bbF^n \rightarrow \bbF$, we use $\zero(f)$ to denote the set $\{x: f(x)=0\}$.
    \item For integers $n, d \in \mathbb{N}$, we use $\mathsf{RM}(n,d)$ to denote the set of polynomials of degree at most $d$ over $\bbF^n$.
    \item For integers $n, k \in \mathbb{N}$ with $n\ge k$, we use $\binom{n}{\le k}$ to denote $\sum_{i=0}^k\binom{n}{i}$.
    \item For a sample oracle $\mathcal{O}$ which return samples in $\bbF^n$, matrix $D\in\bbF^{k\times n}$, we use $D\mathcal{O}$ to denote a new sample oracle which each time returns $D\bx$ where $\bx$ is sampled from $\mathcal{O}$.
    \cadnote{\item For an index set $S$, we use $x_S$ to denote the set $\mathset{x_i : i\in S}$.}
    \cadnote{\item For a set $S$ of vectors, we use $\rank(S)$ to denote $\dim(\sfspan(S))$.}
\end{enumerate}
\subsubsection{Some useful facts regarding polynomials}
We next list some useful facts regarding polynomials over the field $\bbF$. While most of these are easy and standard, we list them here for the sake of completeness. 
\begin{claim}\label{lem:zeroormany}
Let $p$ be a polynomial over $\bbF^n$. If the polynomial $p$ is not identically zero (as a formal expression) and its degree is at most $c$, then
\begin{align*}
    \underset{\bx\sim  \bbF^n}{\mathbb{P}}[p(\bx)\neq 0]\geq 1/2^c.
\end{align*}
\end{claim}
\begin{proof}
The proof is by induction on degree. If $c =0$, then $p$ is identically $1$ and thus the claim follows trivially.

Now, as an inductive hypothesis, assume that the claim is true for all polynomials of degree at most $c-1$. Let $p$ be a polynomial of degree $c$. Since $p$ is not identically zero, there exists $i$ such that $p$ can be expressed as
\begin{equation}~\label{eq:indi-1}
p(x_1,\cdots,x_n) = q(x_1, \ldots, x_{i-1}, x_{i+1}, \ldots, x_n) \cdot x_i +  r(x_1, \ldots, x_{i-1}, x_{i+1}, \ldots, x_n),
\end{equation}
where degree of $q$ is at most $c-1$ and $q$ is not identically zero. The above formulation uses the fact that polynomials over $\bbF$ are multilinear. Observe that any choice of $\bx_{-i} = (\bx_1, \ldots, \bx_{i-1}, \bx_{i+1}, \ldots, \bx_n)$ such that $q(\bx_{-i}) \not =0$,
\begin{equation}~\label{eq:indi-2}
\Pr_{\bx_i \sim \bbF} [p(\bx_1, \ldots, \bx_{i-1}, \bx_i, \bx_{i+1}, \ldots, \bx_n) \not =0] \ge \frac12.
\end{equation}
Now, applying the induction hypothesis on the polynomial $q(x_1, \ldots, x_{i-1}, x_{i+1}, \ldots, x_n)$, we have that
\[
\Pr_{\bx \sim \bbF^n} [q(\bx_1, \ldots, \bx_{i-1},  \bx_{i+1}, \ldots, \bx_n) \not =0] \ge \frac{1}{2^{c-1}}.
\]
Combining this with \eqref{eq:indi-1} and \eqref{eq:indi-2}, we get the claim.

\end{proof}
\begin{claim}~\label{rmk:gassianelimination}
There is an efficient algorithm \textsc{Size-system-polynomial} which given a set of points as input $z_1, \ldots, z_R \in \bbF^n$, determines the size of the set
$T=\setsize{\mathset{p\in \sfRM(n,2): p(z_1)=p(z_2)=\cdots=p(z_r)=0}}$. 
\end{claim}
\begin{proof}
Observe that $p$ can be expressed as linear system of equations (i) where the unknowns are the coefficients of $p$ and (ii) the equations are given by the constraints $\{p(z_i)=0\}_{1 \le i \le R}$. Using Gaussian elimination, we can determine the rank $r$ of this system. Observe that the size of $T$ is just $2^r$, thus proving the claim. 
\end{proof}
\subsubsection{Some useful facts regarding subspaces of $\bbF^n$}
We now list some useful facts about subspaces of $\bbF^n$. 
\begin{claim}\label{lem:spanall}
Let $k,d,n\in\mathbb{N}$ such that $k\geq 100 d$. Let $V\subseteq\bbF^n$ be a subspace of dimension $d$. Let $\bx_1,\cdots, \bx_k$ be $k$ vectors sampled uniformly at random from $V$. Then,
\begin{align}
     \mathbb{P}_{\bx_1,\cdots, \bx_k}[\forall S\subseteq [k] \textrm{ such that }|S|\geq 0.9k,\text{ we have }\sfspan(\bx_S)=V]\geq 1-2^{0.4k}.
\end{align}
\end{claim}
\begin{proof}
We know that there always exist a linear bijection between $V$ and $\bbF^{d}$. \WLOG, we assume $n=d,V=\bbF^d$. Without loss of generality, assume 0.9k is a integer. For a fixed $S$ with $|S|=0.9k$
\begin{align*}
    &\mathbb{P}[\sfspan (\bx_S)=\mathbb{F}_2^d]\\
    &=\prod_{j=0}^{d-1}\penclose{1-2^{-0.9k+j}} 
    &\text{See~ \cite[Equation~(2)]{ferreira2012rank}}\\
    &\geq 1-\sum_{j=0}^{d-1}2^{-0.9k+j}\geq 1- 2^{-0.9k+d}\geq 1-2^{-0.89k}.
\end{align*}
The number of choice of $S$ is at most $\binom{k}{0.1k}\leq (10e)^{0.1k}\leq 2^{0.48k}$. Then the proof is completed by a union bound.
\end{proof}
The next claim says that a union of two proper subspaces of $\bbF^n$ must differ substantially from any subspace of $\bbF^n$. 
\begin{claim}\label{lem:cannotspanall}
Let $S$ be a subspace of $\bbF^n$ and of dimension $d$. Let $U,V\subsetneq S$ be two proper subspaces. Then $|S\backslash(U \cup V)|\ge 2^{d-2}$.
\end{claim}
\begin{proof}
Notice that the size of subspace in $\bbF$ is always a power of 2. There are two cases:\newline
Case 1: $\dim(U)=\dim(V)=d-1$. \newline
Observe that $\dim(U\cap V)\ge d-2$ and hence $|U \cup V|=|U|+|V|-|U\cap V|\le 3\cdot 2^{d-2}$.\newline
Case 2: At least one of $\dim(U)$ or $\dim(V) \le d-2$. \newline
In this case, $|U \cup V|\le |U|+|V|\le 2^{d-1}+2^{d-2}\le 3\cdot 2^{d-2}$.
Thus, in either case, $|U \cup V| \le 3\cdot 2^{d-2}$ which implies that $|S\backslash(U \cup V)|\ge 2^{d-2}$.
\end{proof}

\begin{claim}\label{lem:anotherviewofrandomprojection}
Let $b_1,\cdots,b_t\in\bbF^n$ be linearly independent. Sample $\bfM\in\bbF^{m\times n}$ uniformly at random. Then $\bfM b_1,\cdots, \bfM b_t$ are independent and identically distributed. In other words, the joint distribution of $\bfM b_1,\cdots, \bfM b_t$ is the uniform distribution over $\bbF^{m \times t}$.
\end{claim}
\begin{proof}
Let us first add vectors $b_{t+1}, \ldots, b_n$ such that $\{b_1, \ldots, b_n\}$ is a basis of $\bbF^n$. Let $ B$ be the matrix whose $i^{th}$ column is $b_i$. 
Now, observe that the map $\Psi: \bbF^{m \times n} \rightarrow \bbF^{m \times n}$ defined as $\Psi: M \mapsto M \cdot B$ is a bijection.  Thus, if the random variable $\bfM$ is uniform over $\bbF^{m\times n}$, then so is $\bfM \cdot B$. Consequently, the first $t$ columns of $\bfM \cdot B$, namely, $\bfM b_1, \ldots, \bfM b_t$ are independent and identically distributed. 

\end{proof}
The following theorem gives a hypothesis testing routine for mixtures of subspaces over $\bbF^n$. The proof of this theorem is deferred to Appendix~\ref{apd:uniquedistribution}.
\begin{restatable}{theorem}{hypotest}\thlabel{thm:hypotest:main}
Let $\bD$ be a distribution of a mixture of two incomparable subspaces $A, B \subseteq \bbF^n$ with mixing weights $w_A, w_B \ge w_0$. Let $\mathset{A_j,B_j}_{j=1}^N$ be a collection of $N$ sets of hypothesis with the property that there exists $i$ such that $\mathset{A_i,B_i}=\mathset{A,B}$. There is an \chnote{algorithm \algchoosetherighthypo{}} which is given a confidence parameter $\delta$, $w_0$, $\mathset{A_j,B_j}_{j=1}^N$ and a sampler for $\bD$. Every subspace of $\mathset{A_j,B_j}_{j=1}^N$ will be represented by a basis of that subspace, and the algorithm will have the access to the basis. This algorithm has the following behavior, 
\begin{enumerate}
    \item It runs in $\poly(N,1/w_0)\log(1/\delta)$ time.
    \item With the probability $1-\delta$ outputs the index $i$ such that $\mathset{A_i,B_i}=\mathset{A,B}$.
\end{enumerate}
\end{restatable}

\section{Testing Comparability of the Subspaces}
In this section, the main goal is to prove Theorem~\ref{thm:test_comparability} 
(restated below for the convenience of the reader). We recall that Theorem~\ref{thm:test_comparability} gives an efficient
algorithm which given samples from a mixture of two subspaces $U, V$, decides whether $U$ and $V$ are comparable. This result in turn is an important piece in our 
subspace recovery algorithm in the ``incomparable" case. The algorithm \textsc{Test-comparability} is described in Figure~\ref{alg:test_comparability}. 
\restatetestcomparability*

The main idea of the algorithm is the following. First we take a few samples from the mixture to get $\sfspan(U \cup V)$. By dimension reduction, it suffices to deal with the case $\sfspan(U\cup V)=\bbF^n$. The crucial property we use is the following: If $\sfspan(U\cup V)=\bbF^n$, $U,V$ are incomparable iff there exists non-zero $p\in \sfRM(n,2)$ such that $p$ vanishes on the entire set $U \cup V$. The proof of Theorem~\ref{thm:test_comparability} is deferred to the end of the section -- to start, we prove some auxiliary lemmas. 

\begin{algorithm2e}
\caption{\sc Test-Comparability}
\label{alg:test_comparability}
\KwIn{
\\
$n$ -- ambient dimension \\
$\sampleouv$ -- oracle for random samples from mixture of subspaces. \\
$\wnaught$ -- lower bound of two mixture weights. \\
}
\KwOut{True (if comparable) or False (if incomparable)}
Set $t=16n/(\wnaught^2)$\;
Sample $\bx_1,\cdots, \bx_t$ from $\sampleouv$\;
Set $S=\sfspan(\bx_1,\cdots,\bx_t),v=dim(S)$\;
Find $y_1,\cdots,y_v$ such that they form a basis of $S=\sfspan(\bx_1,\cdots,\bx_t)$.\;
Find a matrix $D\in\bbF^{v\times n}$ such that $D y_i=e_i$ for all $i$, where $e_i$ is the $i$th element of the standard basis of $\bbF^v$.\;
Set $\mathcal{O}'=D \sampleouv = \Dsampleouv$\cnote{This is explained in prelims:useful notation}\;
Set $r=8n^2/\wnaught$\;
Sample $\bz_1,\cdots, \bz_r$ from $\mathcal{O}'=\Dsampleouv$\;
Use algorithm \textsc{Size-System-Polynomial} to compute $T=\setsize{\mathset{p\in \sfRM(v,2): p(\bz_1)=p(\bz_2)=\cdots=p(\bz_r)=0}}$\;\tcp{See \cref{rmk:gassianelimination}}.
\uIf{$T=1$}
{\Return{\True}\;}
\Else
{\Return{\False}\;}
\end{algorithm2e}


\begin{claim}\label{lem:spanofunion}
Assume $s\ge 8n/\wnaught$. Let $\bx_1,\bx_2,\cdots, \bx_s$ be sampled from a mixture of two subspaces $U, V \subseteq \bbF^n$(potentially comparable) of dimension at most $d$ with mixing weights $w_U, w_V \ge \wnaught$. Then, with probability at least $1-\exp(-s \wnaught^2/32)$, 
$\sfspan(\bx_1,\cdots,\bx_s)=\sfspan(U \cup V)$.
\end{claim}
\begin{proof}
For fixed $x_1,\cdots, x_i$ such that $\sfspan(x_1,\cdots,x_i)\subsetneq \sfspan(U \cup V)$, we will show
\begin{align}~\label{eq:numpy1}
    \bbP_{\bx_{i+1}}[\bx_{i+1}\notin \sfspan(x_1,\cdots, x_i)]\ge \wnaught/2.
\end{align}
Define $W=\sfspan(x_1,\cdots,x_i)$. By our assumption, either $U\nsubseteq W$ or $V\nsubseteq W$. Let us assume that it is the former (the other case is 
symmetric). Under this assumption, $U\cap W$ is a proper subset of $U$. \cadnote{Since both are linear subspaces and the size of any linear space over $\bbF$ is always a power of $2$, $|U\cap W|\le 0.5 |U|$. }Hence
\begin{align*}
    \bbP[\bx_{i+1}\in U\backslash W]\ge w_U\frac{|U\backslash W|}{|U|}\ge \wnaught\cdot 0.5.
\end{align*}
In other words, $rank(x_1,\cdots, \bx_{i+1})=rank(x_1,\cdots, x_i)+1$ will hold with probability at least $\wnaught/2$, thus proving \eqref{eq:numpy1}. Define $\by_i=\rank(\bx_1,\cdots,\bx_i)-\rank(\bx_1,\cdots,\bx_{i-1})$, then $\by_1,\cdots,\by_s$ satisfy the condition of \thref{lem:azumadependentchernoff} with $\gamma=\wnaught/2, d=\rank(U\cup V),k=s$. \cref{lem:spanofunion} now follows by applying \thref{lem:azumadependentchernoff}.
\end{proof}
The next (easy) claim says that suppose the  distribution $\bZ$ (over $\bbF^d$)  is \emph{not too concentrated on any single element}. Then, a randomly chosen set 
of size roughly quadratic in $d$ is a \emph{hitting set} for quadratic polynomials over $\bbF^d$. In other words, any non-zero element of $\sfRM(d,2)$ is non-zero 
on at least one element of this set. 
\begin{claim}\label{lem:onlyzerosurvive}
Let $\bZ$ be a distribution over $\bbF^d$ such that the probability weight of every element is at least $w^*/2^d$. Let $\bx_1,\bx_2, \dots, \bx_t$ be independent sampled from $\bZ$. Then, we have
\begin{align*}
    \bbP\Big[ \forall q \in \sfRM(d,2)\setminus \{0\}, \exists j \in [t] \text{ s.t. } q(\bx_j) \ne 0 \Big]\ge 1-\exp\left(-t w^*/4+\binom{d}{\le 2}\log 2 \right).
\end{align*}
\end{claim}
\begin{proof}
Fix $q\in \sfRM(d,2)$ such that $q\neq 0$. By \cref{lem:zeroormany},
\begin{align*}
    \bbP_{\bx\sim_u \bbF^d}[q(\bx)=1]\ge 1/4.
\end{align*}
As a consequence,
\begin{align*}
    \bbP_{\bx\sim Z}[q(\bx)=0]\le 1-\frac{w^*}{4}.
\end{align*}
Hence
\begin{align*}
    \bbP[q(\bx_1)=\cdots=q(\bx_t)=0]\le (1-w^*/4)^t\le \exp(-t w^*/4).
\end{align*}
\chnote{Notice that $|\sfRM(d,2)|=2^{\binom{d}{\le 2}}$}. Using the union bound, we get the claim. 
\end{proof}
\annnote{Aidao, we discussed this. Please do not use Lemma for every intermediate proposition. Lemma should be reserved for main technical statements which are not standalone theorems. Please change them to claims etc.}



We are now ready to finish the proof of Theorem~\ref{thm:test_comparability}.

\begin{proofof}{Theorem~\ref{thm:test_comparability}}
\WLOG, we assume $\delta=0.1$, since we can always boost the probability \chnote{at a multiplicative cost of $\log(1/\delta)$.} 
By \cref{lem:spanofunion}, we know that $S=\sfspan(U\cup V) $ (defined in Step~3 of the algorithm) with probability $0.999$. Henceforth, we assume that $S=\sfspan(U\cup V) $ holds. 

By definition, $D$ (defined in Step~5 of the algorithm) is a linear bijection between $S$ and $\bbF^v$. Hence $DU,DV$ are incomparable if and only if $U,V$ are incomparable. Now observe that, $\mathcal{O}'=\Dsampleouv$ will give samples from mixture of two subspaces $DU, DV$ with mixing weights $w_U, w_V \ge \wnaught$. Notice that $\sfspan(DU \cup DV)=\bbF^v$. We divide the rest of the analysis into two cases.
\newline
Case 1: $DU,DV$ are comparable.
\newline
We have $DU=\bbF^v$ or $DV=\bbF^v$. By \cref{lem:onlyzerosurvive}, with probability $0.999$, there will only be one polynomial (the zero polynomial) in
the set $\mathset{p\in \sfRM(v,2): p(\bz_1)=p(\bz_2)=\cdots=p(\bz_r)=0}$. 
In this case, $T=1$. Thus, overall, with probability $0.998$, algorithm returns the correct answer in this case. 
\newline
Case 2: $DU,DV$ are incomparable.
\newline
In this case, $dim(DU)\le v-1$ (and $dim(DV)\le v-1$). Thus, there exists non-zero vector $b_U$ (resp. $b_V$) such that $\iprod{b_U,DU}=\mathset{0}$ (resp. $\iprod{b_V,DV}=\mathset{0}$). Now, consider the non-zero polynomial $p(x)=\iprod{b_U,x}\iprod{b_V,x}$. By definition it satisfies $p(DU \cup DV)=\mathset{0}$. 
Thus, in this case, the set $\mathset{p\in \sfRM(v,2): p(\bz_1)=p(\bz_2)=\cdots=p(\bz_r)=0}$ has at least two elements. Thus, overall, with probability $0.999$, the algorithm returns the correct answer in this case. 
\end{proofof}

\section{Learning Mixtures of Incomparable Subspaces}

\newa{
In this section, we give a polynomial time algorithm (Algorithm~\ref{alg:incpr_subspace}: \algincprrcv{}) for recovering the subspaces $\Amax, \Amin$ when given access to samples from a mixture of two subspaces that are incomparable. We prove the following theorem.  

}

\restatefastincprrecovery*

\newa{The main idea is a new procedure for {\em dimension reduction} that reduces the subspace clustering problem to $O(1)$ dimensions. We will construct a linear map $M\in\bbF^{10\times n}$ such that after projecting using $M$, the subspaces obtained $M\Amax=\mathset{Mx : x \in \Amax}$ and $M\Amin= \mathset{Mx: x \in \Amin}$ are incomparable. The construction of $M$ involves multiple rounds. In each round, we use Algorithm {\sc Test-Comparability} (and Theorem~\ref{thm:test_comparability}) as a black-box, and find a projection that brings down the dimension by one with high probability, while maintaining incomparability of the subspaces. Once we recover the subspaces $M\Amax, M\Amin$ in $O(1)$ dimensions (\cadnote{using a brute force algorithm: enumerate all possible pairs of subspace, then use \thref{thm:hypotest:main}}), we can then recover the original subspaces $\Amax, \Amin$ by considering samples in $\Amax \cup \Amin$ which are not mapped to $M \Amax \cap M \Amin$ by $M$.  
We defer the proof of Theorem~\ref{thm:intro:main} to the end of section.
}

\begin{algorithm2e}
\caption{\algincprrcv{}}
\label{alg:incpr_subspace}
\KwIn{\\
$n$ -- ambient dimension. \\
$\sampleo$ -- oracle for random samples from mixture of subspaces. \\
$\wnaught$ -- lower bound of two mixture weights. \\ \anote{What is the output?}}
\KwOut{two subspaces.}
$M$=\textsc{Find-A-Good-Projector}($n,\sampleo,\wnaught$)\;
Use brute force to solve \algincprrcv{}($10,M\sampleo,\wnaught$), let $U,V$ be the output \;
Set $t=100n/\wnaught$\;
Sample $\bx_1,\cdots,\bx_t$ from $\sampleo$\;
\Return{$\sfspan(\mathset{\bx_i: M\bx_i\notin V}),\sfspan(\mathset{\bx_i:M \bx_i\notin U})$}\;
\end{algorithm2e}

\newa{ The following lemma is crucial in establishing Theorem~\ref{thm:intro:main}. The lemma proves that with high probability, Algorithm {\sc Find-A-Good-Projector} (Algorithm~\ref{alg:find_projector}) reduces the dimension to $r=10$ while {\em preserving the incomparability} of the subspaces. If $M$ is randomly chosen from $\bbF^{10\times n }$, then $M \Amin \subseteq M \Amax$ since $M \Amax$ collapses to $\bbF^{10}$ with high probability. Algorithm {\sc Find-A-Good-Projector} instead proceeds in multiple rounds, and reduces the dimension one per round. If the projector $\bfM'$ is chosen uniformly at random from $\bbF^{(n-1)\times n}$, with constant probability $\bfM'\Amax, \bfM'\Amin \in \bbF^{n-1}$ remain incomparable. We can now use Algorithm {\sc Test-Comparability} (and Theorem~\ref{thm:test_comparability}) to boost the success probability in each round by repeatedly sampling $M'$ and rejecting it if the resulting subspaces are comparable.     
 }

\begin{lemma}\label{lem:findagoodprojector}
Given samples from a mixture of two incomparable subspaces $\Amax, \Amin \subseteq \bbF^n$ with mixing weights $\wfirst, \wsecond \ge \wnaught$. There exists $M\in \bbF^{10\times n}$ such that $M\Amax,M\Amin$ are incomparable subspaces. Moreover, there is an algorithm \algfindagoodprojector{} that runs in time $1/\wnaught\cdot \text{poly}(n)$ and find such a $M$ with probability at least $0.999$.
\end{lemma}
\begin{algorithm2e}
\caption{\sc Find-A-Good-Projector}
\label{alg:find_projector}
\KwIn{\\
$n$ -- ambient dimension \\
$\sampleo$ -- oracle for random samples from mixture of subspaces. \\
$\wnaught$ -- lower bound of two mixture weights. \\
}
\KwOut{
a matrix $M\in\bbF^{10\times n}$.
}
Set $M=I_{n}$, where $I_{n}\in \bbF^{n\times n}$ is the identity matrix\;
\For{$i=n;i> 10;i=i-1$}
{
Sample $\bfT\in\bbF^{(i-1)\times i}$ uniformly at random\;
\While{\algtestcompr{}\newa{$(i,\bfT M\sampleo,\wnaught,1/n^2)$} \tcp{the last parameter is the failure probability we want.}} 
{
Sample $\bfT\in\bbF^{(i-1)\times i}$ uniformly at random\; \anote{Changed from $T \in \bbF^{i \times (i-1)}$}
}
$M=\bfT M$\;
}
\Return{$M$}\;
\end{algorithm2e}
\anote{9/29: Again, this $TMO(\Amax,\Amin,\dots,\dots)$ is so hard to read/ understand. Can you please elaborate on this as in Algorithm 1?}
\begin{proof}
\newa{ We now show that Algorithm {\sc Find-A-Good-Projector} runs in polynomial time and finds a required projector $M$ with high probability. 
Observe that from Theorem~\ref{thm:test_comparability}, every call of \algtestcompr{} (in step 4 of Algorithm~\ref{alg:find_projector}) fails with probability at most $\delta = O(1/n^2)$. We will prove that at any iteration $i \in \mathset{n,n-1,\dots,11}$, a randomly chosen matrix $\bfT \in \bbF^{(i-1) \times i}$ (in step 3) succeeds with constant probability in preserving the incomparability of the subspaces. This ensures that it will suffice to sample $O(\log n)$ many random $T$ per round before we succeed in that round (and hence $O(n \log n)$ overall).       

}

Fix an iteration $i \in \mathset{n,n-1,\dots,11}$, and let $M \in \bbF^{i \times n}$ be the current projector. Let $U:=M\Amax, V:=M\Amin$, and assume $U, V$ are incomparable. We show the following claim. 

\noindent {\em {\bf Claim:}
For a random $\bfT \in \bbF^{(i-1)\times i}$ chosen in step 3,}
\begin{align}
    \bbP_\bfT[\bfT U,\bfT V\text{ are incomparable}]\ge 9/128. \label{eq:claim:proof}
\end{align}

\noindent We now prove the claim by considering two cases depending on the rank of $U \cup V$ i.e., the dimension of the span of $U \cup V$.

\vspace{4pt}
\noindent {\bf Case 1: $rank(U \cup V)\le i-1$.}
\newline
Let $v=rank(U\cup V)$ and  $b_1,\cdots, b_v$ be a basis of $\sfspan(U \cup V)$. By \cref{lem:anotherviewofrandomprojection}, $\bfT b_1,\cdots, \bfT b_v$ can be viewed as being sampled independently from $\bbF^{i-1}$. A uniformly random matrix from $\bbF^{(i-1)\times(i-1)}$ is full-rank with probability at least $\prod_{j\ge 1}(1-2^{-j}) \ge 1/4$. Hence,
\begin{align*}
    \bbP[\bfT b_1,\cdots, \bfT b_v\text{ are linearly independent}]\ge 1/4.
\end{align*}
When $\bfT b_1,\cdots, \bfT b_v\text{ are linearly independent}$, $\bfT U,\bfT V\text{ are incomparable}$ as required. This establishes \eqref{eq:claim:proof} in Case 1.

\noindent {\bf Case 2: $rank(U \cup V)= i$.}
\newline
Let $b_1,\dots, b_{dim(U\cap V)}$ be a basis of $U\cap V$. We extend the basis such that
\newline
$b_1,\dots, b_{dim(U\cap V)},c_1,\dots, c_{dim(U)-dim(U\cap V)}$ is a basis of $U$, and 
similarly we extend the basis \\
so that $b_1,\dots, b_{dim(U\cap V)},d_1,\dots, d_{dim(V)-dim(U\cap V)}$ is a basis of $V$. Observe that
\newline
$b_1,\dots, b_{dim(U\cap V)},c_1,\dots, c_{dim(U)-dim(U\cap V)},d_1,\dots, d_{dim(V)-dim(U\cap V)}$ is a basis of $\sfspan(U\cup V)$. Reorder this basis to get $a_1,\dots, a_i$ such that $a_{i-1}=c_1,a_{i}=d_1$. Let $\bft_j$ denote $\bfT a_j$. By \cref{lem:anotherviewofrandomprojection}, $\bft_1,\cdots,\bft_i$ are independent and identically distributed. Let $\mathcal{E}$ be the event 
\begin{equation*}
\mathcal{E}= \begin{cases}
\bft_j\notin \sfspan(\bft_1,\cdots,\bft_{j-1}) & \forall 1\le j\le i-3\\
\bft_{i-2}\in \sfspan(\bft_1,\cdots,\bft_{i-3})\\
\bft_{i-1}\notin \sfspan(\bft_1,\cdots,\bft_{i-2})\\
\bft_{i}\notin \sfspan(\bft_1,\cdots,\bft_{i-1})
\end{cases}
\end{equation*}
Then,
\begin{align*}
    \bbP_\bfT[\mathcal{E}]=(\prod_{j=1}^{i-3}(1-2^{j-1}/2^{i-1}))\cdot 1/4 \cdot 3/4 \cdot 1/2\ge 3/4\cdot 3/32=9/128.
\end{align*}
Condition on $\mathcal{E}$. We now show that $\bfT U, \bfT V$ are incomparable as required. We will show $\bfT U\nsubseteq \bfT V$, the other direction is similar. By definition $\bft_{i-1}=\bfT a_{i-1}=Tc_1\in TU$, and $\bft_{i-1}\notin \sfspan(\bft_1,\bft_2,\cdots, \bft_{i-2},\bft_{i})$. However $\bfT V\subseteq \sfspan(\bft_1,\bft_2,\cdots, \bft_{i-2},\bft_{i})$, hence $\bft_{i-1}\notin \bfT V$, $\bfT U\nsubseteq \bfT V$. This establishes \eqref{eq:claim:proof}. Hence the lemma follows. 
\end{proof}

\newa{The following lemma shows that a few samples drawn uniformly from $S \setminus T$ suffice to recover $S$ with high probability. This will allow us to recover $\Amax$ and $\Amin$ after clustering the points in $M\Amax \cup M\Amin$.}

\begin{lemma}\label{lem:minuspropersubspace_spanall}
Let $S$ be a subspace of $\bbF^n$ and of dimension $d$. Let $T$ be a proper subspace of $S$. Let $t\ge 8n$ be a integer. $\bx_1,\cdots,\bx_t$ are independently uniformly sampled from $S\backslash T$. Then,
\begin{align*}
    \bbP[\sfspan(\bx_1,\cdots,\bx_t)=S ]\ge 1-e^{-t/128}.
\end{align*}
\end{lemma}
\begin{proof}
Let $V\subsetneq S$ be a fixed subspace. Then by \cref{lem:cannotspanall}, $|S\backslash (T\cup V)| \ge 2^{d-2}$, which is at least $1/4$ of $|S|$. We have
\begin{align*}
    \bbP_{\bx\sim_u S\backslash T}[\bx\notin V]\ge 1/4.
\end{align*}
In other words, if $\sfspan(\bx_1,\cdots,\bx_k)\neq S$, then $rank(\bx_1,\cdots, \bx_{k+1})=rank(\bx_1,\cdots\bx_k)+1$ will hold with probability at least $1/4$. Define the random variables $\by_i=\rank(\bx_1,\cdots,\bx_i)-\rank(\bx_1,\cdots,\bx_{i-1})$ for $i \in \mathset{1,2,\dots, t}$. Note that $\by_1,\cdots,\by_t$ are not quite independent (since the probability the rank increases at step $i$ depends on the random choices of $\bx_1, \dots, \bx_{i-1}$ in previous iterations).  But they satisfy the condition of \thref{lem:azumadependentchernoff} with $\gamma=1/4, d=\dim(S),k=t$.
The proof is completed after applying \thref{lem:azumadependentchernoff}.
\end{proof}

We are now ready to complete the proof of Theorem~\ref{thm:incpr_rcv}.

\begin{proofof}{Theorem~\ref{thm:incpr_rcv}}
\WLOG, we assume $\delta=0.1$, since we can always boost the probability at a multiplicative cost of $\log(1/\delta)$.
By Lemma~\ref{lem:findagoodprojector}, $M$ satisfies the property that $M\Amax,M\Amin$ are incomparable with high probability (probability at least $0.999$, say). 
Moreover assuming $M\Amax,M\Amin$ are incomparable, the brute force algorithm will return them with high probability.

Let $U=M\Amax,V=M\Amin$.
We will show that $\sfspan(\mathset{\bx_i: M\bx_i\notin V}=\Amax$ with probability $0.998$. 
Observe that $W=\mathset{x\in \Amax: Mx\in  M\Amin}$ is a proper subspace of $\Amax$. Hence if $\bx$ is drawn uniformly from $\Amax$, $\bx$ will not in $W$ with probability at least $1/2$. By Chernoff bound, we expect to see at least $20n$ samples in $\mathset{\bx_i: M\bx_i\notin V}$ with probability $0.999$ and all these samples can be viewed as uniformly drawn from $\Amax\backslash W$. By Lemma~\ref{lem:minuspropersubspace_spanall}, $\sfspan(\mathset{\bx_i: M\bx_i\notin M\Amin}=\Amax$ with probability $0.998$. A similar argument shows that the algorithm also recovers $\Amin$ with high probability. Finally, after recovering $\Amax, \Amin$ it is also easy to estimate the weights $w_0,w_1$ to inverse polynomial accuracy (see Remark~\ref{rem:weights}). 
\end{proofof}

\section{Mixtures of two subspaces with signficant dimension difference}
\newcommand{\monom}[1]{\phi_\ell(#1)}

In this section, we prove Theorem~\ref{thm:dversusalphad} (restated below for convenience of the reader) which shows that there is a computationally efficient algorithm for learning a mixture of two subspaces with significantly different dimensions. Note that the following theorem does {\em not} assume that the two subspaces are incomparable.  

\restatedversusalphad*
The algorithm \textsc{recover-subspace-large-diff} is described in Figure~\ref{alg:constant_relative_dimension}. Before proving Theorem~\ref{thm:dversusalphad}, we will make some simplifying assumptions (with their justifications given below) followed by some useful notation. 
\chnote{
\begin{remark}
\WLOG, we can assume 
\begin{enumerate}
    \item $n=\dmax$. \newd{This is because we can first use Theorem~\ref{thm:test_comparability} to test whether the underlying subspaces are incomparable. If they are incomparable,  we can use Theorem~\ref{thm:incpr_rcv} to recover the subspaces. If not, we can take $O(n/\wnaught)$ samples from the mixture to get 
    $\sfspan(\Amax\cup \Amin)$ with high probability (see \chnote{\cref{lem:spanofunion}}). We can then construct a linear bijection, say $D$, between $\sfspan(\Amax\cup \Amin)$ and $\bbF^{\dmax}$. Applying the map $D$ to every sample from the mixture, we can now assume that $n=\dmax$.}
    \item The algorithm knows $\dmax,\dmin$. This is because we can enumerate all the possible values of $\dmax,\dmin$ and run the algorithm \alglargediff{} to get a list of candidate hypothesis. We can then use the hypothesis testing algorithm in Theorem~\ref{thm:hypotest:main} to identify the correct one with high probability.
    \item We set $\delta=0.1$. This is because we can always boost the success probability of our algorithm at a multiplicative cost of $\log(1/\delta)$.
    \item \chnote{$\dmax$ is at least a sufficiently large constant (which only depends on $\wnaught$). Otherwise, we can always apply a brute force algorithm to recover the subspaces.}
\end{enumerate}
\end{remark}

\noindent {\bf Notation.} 
\begin{enumerate}
    \item We will use $\monom{x} \in \bbF^{\binom{n}{\le \ell}}$ to represent the vector consisting of all the monomials of degree at most $\ell$ on $x$, including the constant term. As an example, when $\ell=2$ and $n=2$,  we have $\monom{x}=(1, x_1, x_2, x_1 x_2)$
-- note that because the underlying field is $\mathbb{F}_2$, all the monomials are multilinear. \chnote{We will use $\monom{A}$ to denote $\mathset{\monom{x}: x\in A}$. $\monom{A}$ is a set of vectors in $\bbF^{\binom{n}{\le\ell}}$.}
    \item We define $t:= \dmax - \dmin = (1-\alpha)\dmax$ to denote the difference between the dimensions of the underlying subspaces $\Amax$ and $\Amin$. 
    \item For a sequence of vector $x_1,x_2,\cdots,x_k$, we define $x_{-i}:=\mathset{x_j:j\neq i}$.
    \item Let us denote by $y_i:= \monom{x_i}$.
\end{enumerate}

Finally, we note that for any subspace $V$ of dimension $d$ over $\bbF$, $rank(\monom{V})=\binom{d}{\le\ell}$.

} 


\begin{algorithm2e}
\caption{\alglargediff}
\label{alg:constant_relative_dimension}
\KwIn{\\
$\dmax$ -- dimension of the larger subspace \\
$\alpha \le 1$ -- ratio of the dimensions of two subspaces \\
$\sampleo$ -- oracle for random samples from mixture of subspaces. \\
$\wnaught$ -- minimum of two mixture weights. \\
}\KwOut{two subspaces $U,V$.}
 Set $\ell=\frac{2\log(100/\wnaught)}{1-\alpha}$\; 
 Use $\sampleo$ to sample $m=\binom{\dmax}{\leq\ell}$ vectors $\bx_1, \bx_2, \cdots, \bx_m$\;
 Let S be the set of all $i\in[m]$ such that $\by_i \coloneqq \monom{\bx_i}$ can be expressed as linear combination of $\{ \monom{\bx_j}: j\neq i\}$\;
\Return{$U=\sfspan(\{\bx_i: i\in S\}),V=\sfspan(\{\bx_i: \bx_i\notin U\}) $}\;
\end{algorithm2e}
We start with the following crucial lemma from \cite{ben2012random} (stated below). An equivalent version was also proven in  \cite[Theorem 1.5]{keevash2005set}.
\begin{lemma}[Lemma 4, \cite{ben2012random}]\label{lem:behl2012lemma4}
Let $x_1,x_2,\cdots, x_R$ be $R=2^r$ distinct points in $\mathbb{F}_2^n$. Consider the linear space of degree $d$ polynomials restricted to these points; that is, the space
\begin{align*}
    \mathset{(p(x_1),\cdots,p(x_R)):p\in\sfRM(n,d)}.
\end{align*}
The linear dimension of this space is at least $\binom{r}{\le d}$.
\end{lemma}
As an easy corollary, we have the following claim. 
\annote{Aidao, please use serif font for notations such as span and dim.}
\begin{lemma}
\label{lem:newbehl2012lemma4}
Let $x_1,x_2,\cdots, x_R$ be distinct points in $\mathbb{F}_2^n$. If $R\geq 2^r$, then $rank(\{\monom{x_1},\cdots,\monom{x_R}\})\ge \binom{r}{\leq\ell}$.
\end{lemma}
\begin{proof}
\WLOG, we can assume $R=2^r$, since having more points can only increase the rank. Let $t=|\sfRM(n,\ell)|$. Say $\sfRM(n,\ell)=\mathset{p_1,\cdots,p_t}$. Let $A\in\bbF^{t\times R}$ be defined as $A_{i,j}=p_i(x_j)$. Applying Lemma~\ref{lem:behl2012lemma4} with $d=\ell$, we know the row-rank of $A$ is at least $\binom{r}{\le\ell}$. Let $B\in\bbF^{\binom{n}{\le\ell}\times R}$ be the matrix whose $i$th column is $\monom{x_i}$. Since every polynomial is a linear combination of monomials, there exists $C\in\bbF^{t\times \binom{n}{\le\ell}}$ such that $A=CB$, hence $\rank(B)\ge \rank(A)\ge \binom{r}{\le\ell}$.
\end{proof}

\begin{proofof}{Theorem~\ref{thm:dversusalphad}}
Let $\Ifirst$ (resp. $\Isecond$) be the set of all $i$ such that $\bx_i$ was sampled from $\Amax$ (resp. $\Amin$). 
We now define the events  $\calE_1$, $\calE_2$, $\calE_3$ and $\calE_4$ as follows:
\begin{enumerate}
    \item $\calE_1$: $\forall i\in \Ifirst, \bfy_i\notin \sfspan(\mathset{\bfy_{-i}}\cup \monom{\Amin})$
    \item $\calE_2$: $|I_1|\ge 10\binom{\alpha\dmax}{\le\ell}$
    \item $\calE_3$: $\forall T\subseteq I_1$ such that $|T|\ge 0.9|I_1|$, we have $\sfspan(\{\bx_j\}_{j \in T})=A_1 $
    \item $\calE_4$: $\sfspan(\{\bx_j\}_{j \in I_0})=A_0$
\end{enumerate}
Assume $\calE_1,\calE_2,\calE_3,\calE_4$ holds. Note that whenever $\calE_1$ holds, it follows that $S$ (defined in line 3 of \alglargediff{}) is a subset of $I_1$.  We now show that $\Amin$ can be recovered from the span of the samples corresponding to $S$. Now, consider the set $\mathset{\monom{\bx_i}:i\in\Isecond\setminus S}$. By definition, the elements of this set are linearly independent (otherwise, they will belong in $S$). As $\dim(\spann(\monom{A_1})) \le \binom{\alpha\dmax}{\le\ell}$, it follows that $|\mathset{\monom{\bx_i}:i\in\Isecond\setminus S}|\le \binom{\alpha\dmax}{\le\ell}$. 
As $i\mapsto \monom{\bx_i}$ is a injection on $I_1\setminus S$
, it follows that $|\{i\in\Isecond\setminus S\}|\le \binom{\alpha\dmax}{\le\ell}$. Since $\calE_2$ holds, $|I_1\setminus S|\le 0.1 |I_1|$, hence $|S|\ge0.9 |I_1|$. Since $\calE_3$ holds, $\spann\penclose{\mathset{\bx_j}_{j\in S}}=\Amin$. 

We now argue that the algorithm also recovers $\Amax$.  We claim $\mathset{j\in[m]:\bx_j\notin A_1}=I_0$. Fix $j\in I_0$\cadnote{. S}ince $\calE_1$ holds, $\monom{\bx_j}=\by_j\notin \monom{A_1}$, then $\bx_j\notin A_1$. Hence $I_0\subseteq\mathset{j:\bx_j\notin A_1}$. It is not hard to see $\mathset{j:\bx_j\notin A_1}\subseteq I_0$. 
Finally when $\calE_4$ holds, we have $\spann(\mathset{\bx_j:\bx_j\notin A_1})=\spann(\mathset{\bx_j:j\in I_0})=A_0$.



Thus, it remains to show that $\calE_1$, $\calE_2$, $\calE_3$ and $\calE_4$ hold simultaneously with probability $0.99$. 
~\\
{\bf Proof of $\bbP[\calE_1]\ge 0.999$:} First, observe that by definition, $\ell=\frac{2\log(100/\wnaught)}{1-\alpha}$. Using the assumption on $\dmax$ and $\wnaught$, it follows that 
\begin{equation} \label{eq:fun1} 
\ell=\frac{2\log(100/\wnaught)}{1-\alpha}  = O \left( \frac{\sqrt{\dmax}}{\log\dmax} \right); \quad  \dmax \ge \frac{2\ell}{(1-\alpha)}. 
\end{equation} 
From this, applying the constraints on $\dmax$ and $\ell$ from \eqref{eq:fun1}, we get 
\begin{equation}\label{eq:fun2} 
        \left(\frac{\wnaught}{100}\right)^{1/\ell}\ge 1+\frac{1}{\ell}\cdot\log\left(\frac{\wnaught}{100}\right)\ge \frac{(1+\alpha)}{2}\ge \alpha+\frac{\ell}{\dmax}.
\end{equation}
Now, it is not difficult to see that $\binom{\alpha \dmax}{\le \ell} \le \binom{\alpha \dmax + \ell}{\ell}$ --  it easily follows from the combinatorial interpretation of binomial coefficients. Now, using this and \eqref{eq:fun2}, we get 
\begin{equation}\label{eq:tensordimratio} 
     \frac{\binom{\alpha \dmax}{\le\ell}}{\binom{\dmax}{\le\ell}}\le\frac{\binom{\alpha \dmax+\ell}{\ell}}{\binom{\dmax}{\ell}}\le \bigg(\alpha +\frac{\ell}{\dmax} \bigg)^\ell\le\frac{\wnaught}{100}.
\end{equation}
We now have, 
\begin{align}
    & \mathbb{P}\Big[\dim(\sfspan(\mathset{\bfy_{-i}}\cup \monom{\Amin}))\leq (1-0.4\wnaught)\binom{\dmax}{\leq \ell} \Big] && \\
    \geq & \mathbb{P}\Big[\dim(\sfspan(\mathset{\bfy_{-i}}\cup \monom{\Amin})) \leq (1-0.5\wnaught)\binom{\dmax}{\leq \ell}+\binom{\alpha \dmax}{\leq \ell} \Big] \nonumber\\
    &\qqquad \text{using \eqref{eq:tensordimratio},}&& \nonumber \\
    \ge &  \mathbb{P}\Big[\dim(\sfspan(\mathset{\bfy_{-i}})) \leq (1-0.5\wnaught)\binom{\dmax}{\leq \ell} \Big] \nonumber\\
    &\qqquad \text{using $\dim(\sfspan(\monom{\Amin})) = \binom{\alpha \dmax}{\leq \ell}$,}&& \nonumber \\
    \ge & \bbP[|I_0|\le (1-0.5\wnaught)\binom{\dmax}{\leq \ell}]\nonumber\\ 
    &\qqquad \text{using $|I_0| \ge |\mathset{\bfy_{-i}}| \ge \dim(\sfspan(\mathset{\bfy_{-i}}))$,}&& \nonumber\\
    \ge & 1-e^{-\frac{\wnaught^2}{24}\binom{\dmax}{\leq\ell}}\label{eq:fun4}  \\
    &\qqquad \text{from a standard Chernoff bound}.\nonumber&& 
\end{align}
Let us now define the event $\mathcal{B}_i$ as the event that $i \in I_0$ and $\dim(\sfspan(\mathset{\bfy_{-i}}\cup \monom{\Amin}))\leq (1-0.4\wnaught)\binom{\dmax}{\leq \ell} $. Let $r \coloneqq \lceil (1-0.4\wnaught/\ell)\dmax+\ell \rceil$. Using  reasoning  similar to  \eqref{eq:tensordimratio}, we have 
\[
    \frac{\binom{r}{\le \ell}}{\binom{\dmax}{\le\ell}}\ge \frac{\binom{r}{\ell}}{\binom{\dmax+\ell}{\ell}}\ge \left(\frac{r-\ell}{\dmax}\right)^\ell
    \ge \bigg(1-\frac{0.4\wnaught}{\ell}\bigg)^\ell
    \ge 1-0.4\wnaught.
\]
Thus, it follows that if the event $\mathcal{B}_i$ holds, $\dim(\sfspan(\mathset{\bfy_{-i}}\cup \monom{\Amin})) \le \binom{r}{\le \ell}$.
Now, let us define the set $\mathcal{H}_i = \{x \in \bbF^{\dmax}: \monom{x} \in \sfspan(\mathset{\bfy_{-i}}\cup \monom{\Amin})\}$. By Lemma~\ref{lem:newbehl2012lemma4},   we get that $|\mathcal{H}_i| \le 2^{r+1}$. 
Thus, we now have 
\begin{align}
\mathbb{P}[\bfy_i\in \sfspan(\mathset{\bfy_{-i}}\cup \monom{\Amin})|\mathcal{B}_i] = \frac{|\mathcal{H}_i|}{2^{\dmax}}  \le 
\frac{2^{r+1}}{2^{\dmax}} \le 2^{-\frac{0.35\wnaught \dmax}{\ell}}.
\end{align}
Applying the above inequality along with \eqref{eq:fun4}, we get 
\begin{align}
\mathbb{P}[\bfy_i\notin \sfspan(\mathset{\bfy_{-i}}\cup \monom{\Amin})|i\in I_0]\ge 1- 2^{\frac{-0.35\wnaught \dmax}{\ell}}-e^{-\frac{\wnaught^2}{24}\binom{\dmax}{\leq\ell}}\ge 1-2^{\frac{-0.3\wnaught \dmax}{\ell}}.     \label{eq:fun6}
\end{align} 
By taking a union bound, it follows that 
\begin{align}
    \mathbb{P}[\forall i\in \Ifirst, \bfy_i\notin \sfspan(\mathset{\bfy_{-i}}\cup \monom{\Amin})]\ge 1- \binom{\dmax}{\leq\ell}2^{\frac{-0.3\wnaught \dmax}{\ell}}\ge 1- 2^{\frac{-0.2\wnaught \dmax}{\ell}}.\label{eq:fun7}
\end{align}
As we have chosen $\dmax$ to be sufficiently large, the right hand side is at least $0.999$ showing that $\bbP[\calE_1]\ge 0.999$. 
\newline
\newline
{\bf Proof of $\bbP[\calE_2]\ge 0.999$:} This follows from  a straightforward Chernoff bound on the sampling process defining  $I_1$. 
\newline
\newline
{\bf Proof of $\bbP[\calE_3]\ge 0.999$:} This is a direct application of \cref{lem:spanall}. 
\newline
\newline 
{\bf Proof of $\bbP[\calE_4]\ge 0.999$:} This also follows from \cref{lem:spanall}. 



\end{proofof}

\section{Reduction from Learning Noisy Parities} \label{sec:paritieslb}
\anote{9/29: finished editing this section. Aidao, please go over it once, after you're done with Section 6.  }

In this section, we show how the problem of learning a mixture of two (comparable) subspaces captures the notorious hard problem of  {\em learning parity with noise} (LPN).

\newa{
Given $n \in \mathbb{N}$, the $(n,\epsilon)$-{\sc LPN} problem is instantiated by an (unknown) parity function $f: \bbF^n \to \bbF$ and a noise parameter $\epsilon \in (0,1/2)$. 
The samples are generated i.i.d. by a sampling oracle $\mathcal{O}=\mathcal{O}(f,\epsilon)$ as follows. 
First,  $\bfx\sim_u\bbF^n$ is sampled uniformly at random from $\bbF^n$. Then $\mathbf{b}\in\mathset{0,1}$ is sampled such that  $\bbP[\mathbf{b}=0]=1-\epsilon$ and $\bbP[\mathbf{b}=1]=\epsilon$. If $\mathbf{b}=0$, $\mathcal{O}$ outputs $(\bfx,f(\bfx))$ and if $\mathbf{b}=1$, outputs $(\bfx,1-f(\bfx))$. Given samples generated i.i.d. by the sampling oracle $\mathcal{O}(f,\epsilon)$, the goal is to learn the unknown parity function $f$.}

The following simple proposition reduces {\sc LPN} to learning mixtures of (comparable) subspaces in $\bbF^{n+1}$, where the subspaces have dimensions $n+1$ and $n$ respectively. 
\begin{proposition}\label{prop:reduce_LPN_to_mix}
Suppose there exists an algorithm {\sc ALG} that given samples from a mixture of two subspaces $\Amax=\bbF^{n+1}, \Amin\subseteq \bbF^{n+1}$ of dimensions $n+1,n$ respectively, with mixing weights $2\epsilon,1-2\epsilon$, runs in time $T=T(n,\delta)$ and solves this problem with probability $1-\delta$. Then there is an algorithm that solves $(n,\epsilon)$-{\sc LPN} with probability $1-\delta$ and running time $O(T)+\poly(n)$.
\end{proposition}
\begin{proof}
\newa{Consider a sample $(\bx,\by)\in\bbF^{n+1}$ (with $\bx \in \bbF^n$) drawn from a sampling oracle $\mathcal{O}(f,\epsilon)$ for the $(n,\epsilon)$-{\sc LPN} problem.} 
We can view $(\bx, \by)$ as a sample from a mixture of two subspace $\bbF^{n+1}, \Amin\subseteq \bbF^{n+1}$ of dimension $n+1,n$ (respectively) with mixing weights $2\epsilon, (1-2\epsilon)$ as follows. 
\newa{Let $\Amin$ be the subspace of dimension $n$ defined by the linear equation $f(\bx)+\by=0$ over $\bbF$. On the one hand, if $\bfb=1$, then $(\bx,\by) \in \bbF^{n+1}$ does not belong to $\Amin$; it is drawn from $\Amax \setminus \Amin$.  On the other hand when $\bfb=0$, $(\bx,\by) \in \bbF^{n+1}$ lies in the subspace $\Amin$. But this could correspond to a sample drawn from $\Amin$ or to the portion of $\Amax$ that overlaps with $\Amin$ (recall that $\Amin \subset \Amax$ and $|\Amax \cap \Amin|=|\Amax|/2$ in our case). Hence by setting the mixing weights of the subspaces $\Amax=\bbF^{n+1}, \Amin$ to be $2\epsilon, 1-2\epsilon$ respectively, we can view a sample $(\bx,\by)$ drawn from the {\sc LPN} problem as being drawn from the mixture of subspaces $\Amax, \Amin$. 

Our goal is then to recover $\Amax, \Amin$ from i.i.d. samples of the form $(\bx,\by)$ drawn from the {\sc LPN} problem.  If the algorithm {\sc ALG} succeeds in finding $\Amin$, then this provides a parity function $f$ (corresponding to the constraint defining $\Amin$) that satisfies the LPN problem.}
\end{proof}

The next proposition shows that learning mixtures of two subspaces $\Amax, \Amin$ in $\bbF^{n+1}$ where $\Amax=\bbF^{n+1}$ and $\dim(\Amin)=n$ is in fact equivalent to the {\sc LPN} problem. 

\begin{proposition}
\label{prop:reduce_mix_to_LPN}
Suppose there is an algorithm $ALG$ that solves $(n,\epsilon)$-LPN with probability $1-\delta$ and running time $T=T(n,\delta)$. Then, there is an algorithm that given samples from a mixture of two subspaces $\bbF^{n+1}, \Amin\subseteq \bbF^{n+1}$ of dimension $n+1,n$ respectively with mixing weights $2\epsilon,1-2\epsilon$, runs in time $O(nT)+\poly(n)$ and recovers $\Amin$ with probability $1-\delta-\exp(-n)$. 
\end{proposition}
\begin{proof}
\newa{ We start with a simple observation. Suppose  (*)  $x_{i_1}+x_{i_2}+\cdots+x_{i_k}=0$  be the constraint defining subspace $\Amin$, and suppose $j\in\mathset{i_1,i_2,\cdots,i_k}$. Consider the parity 
$$f:\bbF^{\mathset{1,2,\dots,n+1}\setminus \mathset{j}} \to \bbF, \text{ where } f(x)=\sum_{\ell \in \mathset{i_1,i_2, \dots, i_k} \setminus \mathset{j}} x_\ell .$$
On one hand, if $(\bfx_1,\dots,\bfx_{n+1})$ is drawn from $\Amin$ (this is with probability $1-2\epsilon$), then the pair $(\bfx_{-j},\bfx_j)$ satisfies the parity $f$ by definition of $\Amin$. On the other hand, if $(\bfx_1,\dots,\bfx_{n+1})$ is drawn from $\Amax$ (this is with probability $2\epsilon$), it satisfies parity $f$ with probability $1/2$. In total, the parity $f$ is satisfied with probability $1-2\epsilon+ \tfrac{1}{2} (2\epsilon) = 1-\epsilon$. Hence, a sample  $(\bfx_1,\dots,\bfx_{n+1})$ from the mixture of subspaces with weights $2\epsilon, 1-\epsilon$, $(\bfx_{-j},\bfx_j)$ can be viewed as a sample of $(n,\epsilon)$-{\sc LPN} with unknown parity $f$. 

We do not know $\mathset{i_1, i_2, \dots, i_k}$. However we can guess and try out $j=1,\cdots, j=n+1$ and get at most $n+1$ candidate hypothesises. We can then use the well known hypothesis testing result from Proposition~\ref{prop:gnrl_hypo_test} to filter and find the correct subspace $\Amin$ with high probability.}
\end{proof}

\acks{We thank Swastik Kopparty for telling us about the results in \cite{ben2012random}.}

\bibliography{aravind}

\appendix

\section{Hypothesis Test}
\label{apd:uniquedistribution}

In this section we will prove the following theorem.
\hypotest*
We defer the proof to the end of this section.

In order to prove Theorem~\ref{thm:hypotest:main}, we need a fundamental tool from statistics, namely “hypothesis testing for distributions”. There are many equivalent forms of this algorithm — we use the following (convenient) version from \cite{de2014learning}.
\begin{proposition}[Simplified {\cite[Proposition 6]{de2014learning}}]\label{prop:gnrl_hypo_test}
Let $\bD$ be a distribution over $W$ and $\bD_\epsilon=\mathset{\bD_j}_{j=1}^N$ be a collection of $N$ distribution over $W$ with the property that there exists $i\in[N]$ such that $d_{TV}(\bD,\bD_i)\le\epsilon$. There is an algorithm $T^D$ which is given an accuracy parameter $\epsilon$, a confidence parameter $\delta$, and is provided with access to (i) samplers for $\bD$ and $\bD_k$, for all $k\in[N]$ (ii) a evaluation oracle $EVAL_{\bD_k}$, for all $k\in[N]$, which, on input $w\in W$, output the value $\bD_k(w)$. This algorithm has the following behavior: It makes $m=O((1/\epsilon^2)(\log N+\log(1/\delta)))$ draws from $\bD$ and each $\bD_k,k\in[N]$, and $O(m)$ calls to each oracle $EVAL_{\bD_k},k\in[N]$, performs $O(mN^2)$ arithmetic operations, and with probability $1-\delta$ outputs an index $i^*\in[N]$ that satisfies $d_{TV}(\bD,\bD_{i^*})\le 6\epsilon$.
\end{proposition}

\begin{restatable}{definition}{mixdistribution}
$\bD(A,B,w_A,1-w_A)$ is defined as the distribution induced by a mixture of two incomparable subspaces $A, B \subseteq \bbF^n$ of dimension at most $d$ with mixing weights $w_A, 1-w_A$.
\end{restatable}

\begin{lemma}\label{lem:hypo_dtv_seperation}
Let $A,B,C,D$ be 4 subspaces of $\bbF^n$. Suppose $\mathset{A,B}\neq \mathset{C,D}$. Let $\bD_1=\bD(A,B,w_A,1-w_A), \bD_2=\bD(C,D,w_C,1-w_C), w^*=min(w_A,1-w_A,w_C,1-w_C)$. Then $d_{TV}(\bD_1,\bD_2)\ge w^*/8$.
\end{lemma}
\begin{proof}
\WLOG, assume $A$ has largest dimension among all 4 subspaces. We divide the rest of the analysis into a few cases.

\begin{align*}
    \begin{cases}
        Case\ 1: A\ne C \text{ and }A\ne D.
        \\
    A=C\text{ or }A=D.
    \text{ Assume A=C.}\footnotemark
    \begin{cases}
    Case\ 2: A=B\text{ or } A=D.
    \\
    A\ne B\text{ and } A\ne D.
        \begin{cases}
            Case\ 3: A,B\text{ are incomparable.}
            \\
            Case\ 4: A,D\text{ are incomparable.}
            \\
            Case\ 5: B\subsetneq A\text{ and } D\subsetneq A.
        \end{cases}
    \end{cases}
    \end{cases}
\end{align*}
\footnotetext{This is without loss of generality.}
Case 1: 
\newline
In this case, $A\cap C$ and $A\cap D$ are two proper subspace of $A$. By \cref{lem:cannotspanall}, $|A\backslash (C\cup D)|\ge |A|/4$, $d_{TV}(\bD_1,\bD_2)\ge w^*/4$.
\newline
Case 2: 
\newline
\WLOG, assume $A=B$. We have $dim(A)\ge dim(D)$ and $D\ne A$. Hence $A\cap D$ is a proper subspace of $A$. $|(\bD_1-\bD_2)(A\backslash D)|=(1-w_C) |A\backslash D|/|A|\ge w^* \cdot 1/2$.
\newline
Case 3:
\newline
If $B\subseteq D$, we have $B\subsetneq D$. Since $A,B$ are incomparable, $A,D$ are incomparable. $|(\bD_1-\bD_2)(D\backslash (A\cup B)|\ge w^*/4$. If $B\nsubseteq D$, $B\cap D$ is a proper subspace of $B$, $|(\bD_1-\bD_2)(B\backslash (A\cup D)|\ge w^*/4$.
\newline
Case 4: similar to Cases 3.
\newline
Case 5:
\newline
If $|w_A-w_C|\ge w^*/2$, then $|(\bD_1-\bD_2)(A\backslash(B\cup D))|=|w_A-w_C|\cdot|A\backslash(B\cup D))|/ |A|\ge w^*/2 \cdot 1/4$. If $|w_A-w_C|\le w^*/2$, \lowerwlog, assume $dim(B)\ge dim(D)$. Since $B\ne D$, $B\cap D$ is a proper subspace of $B$. $|(\bD_1-\bD_2)(B\backslash D)|=|(w_A-w_C)\cdot|B\backslash D|/ |A|+(1-w_A)|B\backslash D|/ |B| |\ge (1-w_A)|B\backslash D|/|B|-|(w_A-w_C)\cdot|B\backslash D|/ |A| |\ge w^*/2-w^*/2\cdot 1/2=w^*/4$.
\end{proof}

\begin{proof}[Proof of Theorem~\ref{thm:hypotest:main}]
Set $\epsilon=w_0/100,M=\ceil{1/\epsilon},\gamma=(1-w_0)/M$. Let $\bD_\epsilon=\mathset{\bD(A_j,B_j,w_0+k*\gamma,1-w_0-k*\gamma}_{j\in[N],k\in[M]\cup\mathset{0}}$. It is not hard to see that there exist $\bD^*\in \bD_\epsilon$ such that $d_{TV}(\bD^*,\bD)\le \epsilon$. By Proposition~\ref{prop:gnrl_hypo_test}, we can find $\bD'\in \bD_\epsilon$ such that $d_{TV}(\bD',\bD)\le 6\epsilon$ with probability $1-\delta$. Say $\bD'=\bD(A',B',w',1-w')$. We claim $\mathset{A',B'}=\mathset{A,B}$. For a contradiction, suppose it is not true. Then by Lemma~\ref{lem:hypo_dtv_seperation}, $d_{TV}(\bD',\bD)\ge w_0/8 > 6\epsilon$, we derive a contradiction.
\end{proof}

\section{Generalized Chernoff Bound}
\begin{lemma}\thlabel{lem:azumadependentchernoff}
Let $\gamma\in(0,1),d,k\in\mathbb{N}$. Let $\bx_1,\bx_2,\cdots,\bx_k$ be a sequence of random variables such that for all $i\in[k]$
\begin{align*}
    \bbP[(\bx_i=1) \lor (\bx_1+\bx_2+\cdots+\bx_{i-1}\ge d) | \bx_1,\cdots,\bx_{i-1}]\ge \gamma.
\end{align*}
Assume $k\ge 2d/\gamma$.
Then 
\begin{align*}
    \bbP[\bx_1+\cdots+\bx_k\ge d]\ge 1-\exp\penclose{-k\gamma^2/8}.
\end{align*}
\end{lemma}

\begin{proof}
We will use the coupling technique.
Define 
\begin{align*}
    \by_i=
    \begin{cases}
        1&\text{ if } \bx_1+\cdots+\bx_{i-1}\ge d.\\
        \bx_i&\text{ otherwise}.
    \end{cases}
\end{align*}
Then 
\begin{enumerate}
    \item $\bx_1+\cdots+\bx_k\ge d\iff \by_1+\cdots+\by_k\ge d$.
    \item For all $i\in[k]$,$\bbP[\by_i=1 | \by_1,\cdots,\by_{i-1}]\ge \gamma$.
\end{enumerate}
Define a submartingale $\bZ_0,\cdots, \bZ_k$ by $\bZ_0=0$ and $\bZ_j=\sum_{1\le l\le j}\by_l-j\gamma$. Then,
\begin{align*}
    &\bbP[\bx_1+\cdots+\bx_k\ge d]\\
    &=\bbP[\by_1+\cdots+\by_k\ge d] \\
    &=1-\bbP[\by_1+\cdots+\by_k\le d-1]\\
    &\ge 1-\bbP[\bZ_k-\bZ_0\le d-1-k\gamma]  \\
    &\ge 1-\exp\penclose{-\frac{(k\gamma-(d-1))^2}{2k}} &&\text{by Azuma–Hoeffding inequality}\\
    &\ge 1-\exp\penclose{-k\gamma^2/8}.&&\text{by $k\gamma\ge 2d$}\\
\end{align*}
\end{proof}
\end{document}